\documentclass{article}
\usepackage{amssymb}

\usepackage{graphicx}
\usepackage{color}
\usepackage{amsmath}

\newtheorem{theorem}{Theorem}

\newtheorem{example}[theorem]{Example}

\newtheorem{proposition}[theorem]{Proposition}
\newtheorem{remark}[theorem]{Remark}

\newenvironment{proof}[1][Proof]{\textbf{#1.} }{\ \rule{0.5em}{0.5em}}

\def\essinf{\operatorname{ess.\!inf}}
\def\esssup{\operatorname{ess.\!sup}}

\def\bR{\mathbb{R}}

\def\bE{\mathbb{E}}

\def\cM{\mathcal{M}}
\def\cC{\mathcal{C}}

\def\cP{\mathcal{P}}

\def\cQ{\mathcal{Q}}

\def\cD{\mathcal{D}}
\def\cW{\mathcal{W}}
\def\cF{\mathcal{F}}

\begin{document}

\title{Time-consistency of cash-subadditive risk measures}
\author{Elisa Mastrogiacomo \thanks{Dipartimento di Statistica e Metodi Quantitativi. University of Milano-Bicocca, Italy. e-mail: elisa.mastrogiacomo@unimib.it} and Emanuela Rosazza Gianin \thanks{Dipartimento di Statistica e Metodi Quantitativi. University of Milano-Bicocca, Italy. e-mail: emanuela.rosazza1@unimib.it} \thanks{Part of this work was done while the last author was visiting the
ZIF center in Bielefeld, Germany, as member of the Research Group in Robust Finance:
Strategic Power, Knightian Uncertainty, and the Foundations of Economic Policy Advice.
The Financial support, the warm hospitality of ZIF center and the stimulating atmosphere
are gratefully acknowledged.}}
\maketitle

\begin{abstract}
  The main goal of this paper is to investigate under which conditions cash-subadditive convex dynamic risk measures are time-consistent. Proceeding as in Detlefsen and Scandolo \cite{detlef-scandolo} and inspired by their result, we give a dual representation of dynamic cash-subadditive convex risk measures (that can also be seen as particular case of the dual quasiconvex representation). The main result of the paper consists in providing, in the cash-subadditive case, a sufficient condition for strong time-consistency (or recursivity) in terms of a generalized cocycle condition. On one hand, our result can be seen as an extension to cash-subadditive convex dynamic risk measures of Theorem 2.5 in Bion-Nadal \cite{bion-nadal-FS}; on the other hand, it is weaker since strong time-consistency is not fully characterized. Finally, we exploit the relation between different notions of time-consistency.
\end{abstract}

\section{Introduction}

Starting from the seminal work of Artzner et al. \cite{ArDeEbHe99} on coherent risk measures, an increasing attention has been devoted to quantifying the riskiness of financial positions. Coherent risk measures have been introduced and defined axiomatically by Artzner et al. \cite{ArDeEbHe99} and Delbaen \cite{De00}, by imposing a set of axioms that are reasonable (or, better, coherent) from a financial point of view.  Motivated by liquidity arguments,  F\"ollmer
and Schied \cite{follers} and Frittelli and Rosazza Gianin \cite{frittelli-rg} introduced independently the wider class of convex risk measures by replacing the axioms of positive
homogeneity and subadditivity with the weaker axiom of convexity.

While the notions above deal with quantifying now the riskiness of a financial position (static setting), it is more realistic to consider a dynamic setting where the riskiness of the position would be quantified at any time between the current one and a fixed future horizon. For this reason, dynamic coherent and convex risk measures have been introduced and investigated. See, among many others, Artzner et al. \cite{ArDeEbHeKu04}, Bion-Nadal \cite{bion-nadal-FS}, \cite{bion-nadal-SPA}, Cheridito et al. \cite{CDK3}, Delbaen \cite{delb-mstable}, Detlefsen and Scandolo \cite{detlef-scandolo}, F\"ollmer and Penner
\cite{follmer-penner}, Frittelli and Rosazza Gianin \cite{frittelli-rg2}, Kl\"oppel and
Schweizer \cite{kloppel-schweizier} and
Riedel \cite{Ri04}.

A key property, in the dynamic setting, is the notion of time-consistency. Among the different notions of time-consistency introduced and
studied in the literature, the most widely used is the so-called strong time-consistency,
corresponding to recursivity. While time-consistency of dynamic coherent risk measures is strongly related to $m$-stability (or rectangularity) of the set of probability measures appearing in the dual representation of such risk measures (see Delbaen \cite{delb-mstable}), in the dynamic convex case strong time-consistency has been characterized by means of a decomposition property on acceptance sets (see Cheridito et al. \cite{CDK3}) and in terms of a property (called cocycle) on the minimal penalty term (see Bion-Nadal \cite{bion-nadal-FS} in continuous time and F\"{o}llmer and Penner \cite{follmer-penner} in discrete time). Further studies on time-consistency of risk measures can be found in Acciaio and Penner \cite{AcPe11}, Cheridito and Kupper \cite{CK}, Delbaen et al. \cite{DPRG}, Detlefsen and Scandolo \cite{detlef-scandolo}, Drapeau et al. \cite{drapeau-kupper-tangpi-rg}, Kl\"oppel and
Schweizer \cite{kloppel-schweizier}, Riedel \cite{Ri04}, Roorda and Schumacher \cite{RoSch15} and Rosazza Gianin \cite{ros}, among others.

Although in the aforementioned works on static and dynamic risk measures cash-additivity is often assumed, such axiom has been recently discussed by El Karoui and Ravanelli \cite{ELK-rav} (and later on also by Cerreia-Vioglio et al. \cite{CMMM2}, Drapeau and Kupper \cite{drapeau-kupper} and Frittelli and Maggis \cite{fritt-mag-dyn}). As argued by these authors, indeed, cash-additivity is too strong, mainly when dealing with stochastic interest rates or ambiguity over discounting. Motivated by this argument, the wider class of cash-subadditive risk measures has been introduced by El Karoui and Ravanelli \cite{ELK-rav} by replacing cash-additivity with cash-subadditivity.\medskip

In this paper, we focus on dynamic convex cash-subadditive dynamic risk
measures, in the perspective of generalizing the results established in the literature
for dynamic convex cash-additive risk measures.

First, we provide a dual representation of dynamic convex cash-subadditive risk measures by means of a penalty term and of discount factors and by following an approach that is different from the one used by El Karoui and Ravanelli \cite{ELK-rav} for static risk measures.

Second, we prove that a generalized cocycle condition on the penalty term together with suitable conditions on the discount factors and on the set of probability measures guarantees the strong time-consistency of the corresponding dynamic risk measure.

Finally, we discuss and relate the different notions of time-consistency proposed in the literature in the cash-subadditive case. In particular, we study the link between strong, weak and weak* time-consistency. We emphasize that, because of the lack of cash-additivity,
one cannot expect the equivalence between these notions. Although strong time-consistency implies weak time-consistency, we show that the converse is no more true when cash-additivity is replaced by cash-subadditivity.
On the
one hand, our results on time-consistency can be seen as an extension to dynamic convex cash-subadditive risk measures
of  \cite[Theorem 2.5]{bion-nadal-FS} and \cite[Proposition 5]{detlef-scandolo};
on the other hand, they are weaker since we are not able to prove full characterizations.\bigskip

The paper is organized as follows: in Section \ref{sec:assump} we introduce notations and basic assumptions used in the paper; in Section \ref{sec:rr} we review and refine the dual representation results of dynamic
convex risk measures when the cash-additivity assumption
is replaced by cash-subadditivity. The main result of the paper can be found in Section \ref{sec:cocycle}, where we provide sufficient conditions (in terms of a generalized cocycle property for the penalty and of pasting properties on the discount factors and on the set of probability measures) for a dynamic convex cash-subadditive risk measure to be time-consistent. Finally, the link between different notions of time-consistency is considered in Section \ref{sec:properties}. Example and counterexamples which emphasize the differences between the cash-additive and the cash-subadditive case are also considered. Section \ref{conclusion} collects some concluding remarks, while the appendix contains the proofs of the main results.

\section{Notation and initial remarks}\label{sec:assump}

Let $(\Omega, \mathcal{F})$ be a measurable space and let $P$ be a probability measure defined on it.

Denote by $\cM_{1,f} \triangleq \cM_{1,f}(\Omega,\cF)$ the class of all finitely additive set functions $Q: \cF \to [0,1]$ that are normalized to $1$ and by $\cM_1=\cM_{1}(\Omega,\cF)$ the subset of $\cM_{1,f}$ formed by all the $\sigma$-additive elements of $\cM_{1,f}$, that is the class of all probability measures on $(\Omega,\cF)$. Furthermore, $\cM_{s,f}$ (resp. $\cM_{s}$) will denote the set of all finitely additive (resp. $\sigma$-additive) measures $\mu$ on $(\Omega,\cF)$ such that
$0\leq \mu(\Omega)\leq 1$ (called subprobabilities).
\begin{align*}
   \cM_{1}(P)\triangleq \cM_{1}(\Omega,\cF,P) \qquad \textrm{(resp. $\cM_{s}(P)\triangleq \cM_{s}(\Omega,\cF,P))$}
\end{align*}
will denote the set of all $\sigma$-additive probability (resp. subprobability) measures on $(\Omega,\cF)$ that are absolutely continuous with respect to $P$.

With an abuse of notation, in the following $\bE_Q[X]$ will denote the integral of $X$ with respect to $Q\in \cM_{1,f}$.

Notice that any element $\mu\in \cM_{s,f}$ (resp. in $\cM_{s}$) can be decomposed as $\mu(\cdot)=aQ(\cdot)$ for some constant $a\in [0,1]$ and some measure $Q\in \cM_{1,f}$ (resp. in $\cM_{1}$).
If $\mu=0$ then $a=0$ and $Q\in \cM_{1,f}$ (resp. in $\cM_{1}$) is not uniquely identified.
\bigskip

In the following, we will focus on random variables on $L^{\infty}(\Omega, \mathcal{F},P)$, where $L^{\infty}(\Omega, \mathcal{F},P)$ is the space of all essentially bounded random variables on $(\Omega, \mathcal{F},P)$. For simplicity of notations, we will often write $L^{\infty}$ instead of $L^{\infty}(\Omega, \mathcal{F},P)$.
We recall that the topological dual space of $L^{\infty}$ endowed with the $\Vert \cdot \Vert_{\infty}$ is $ba$, while the one of $L^{\infty}$ endowed with the weak$^*$ topology $\sigma(L^{\infty},L^1)$ is $L^1$.
\medskip

Let $T$ be a finite fixed time horizon, let $\mathcal{T}$ be either the set $\{0,1,...,T\}$ (discrete time) or the time interval $[0,T]$ (continuous time) and let $(\mathcal{F}_t)_{t \in \mathcal{T}}$ be a filtration of $\mathcal{F}$ satisfying $\mathcal{F}_0=\{\emptyset; \Omega \}$ and $\mathcal{F}_T=\mathcal{F}$.
\bigskip

We recall that a static risk measure $\rho$ is a functional $\rho:L^{\infty} \to \Bbb R$ satisfying some suitable assumptions. (Static) coherent and convex cash-additive risk measures have been widely discussed and studied in the literature. See Artzner et al. \cite{ArDeEbHe99}, Delbaen \cite{delb}, F\"{o}llmer and Schied \cite{follers}, \cite{follmers-book}, Frittelli and Rosazza Gianin \cite{frittelli-rg}, among many others. In a dynamic setting, a dynamic risk measure has been defined as a family $(\rho_t)_{t \in \mathcal{T}}$ of functionals $\rho_t: L^{\infty}(\mathcal{F}_T) \to L^{\infty}(\mathcal{F}_t)$ taking into account all the information available till time $t$. Similar axioms as in the static case are sometimes imposed to dynamic risk measures (see, among others, Artzner et al. \cite{ArDeEbHeKu04}, Detlefsen and Scandolo \cite{detlef-scandolo}, F\"{o}llmer and Penner \cite{follmer-penner}, Frittelli and Rosazza Gianin \cite{frittelli-rg2}, Kl\"{o}ppel and Schweizer \cite{kloppel-schweizier} and Riedel \cite{Ri04}). Here below a list of the main ones:
\smallskip

\noindent - convexity: $\rho_t( \alpha X + (1- \alpha) Y) \leq \alpha \rho_t(X)+ (1 - \alpha) \rho_t(Y)$ for any $t \in \mathcal{T}$, $X, Y \in L^{\infty}(\mathcal{F}_T)$, $\alpha \in [0,1]$; \smallskip

\noindent - monotonicity: $X \leq Y$, $P$-a.s., implies that $\rho_t(X) \geq \rho_t(Y)$ for any $t \in \mathcal{T}$; \smallskip

\noindent - continuity from above (respectively below): $X_n \downarrow _n  X$ (resp. $X_n \uparrow _n  X$) implies that $\lim_n \rho_t(X_n) =\rho_t(X)$ for any $t \in \mathcal{T}$; \smallskip

\noindent - cash-additivity: for any $t \in \mathcal{T}$,  $\rho_t (X+m_t) = \rho_t (X)-m_t$ for any $X \in L^{\infty}(\mathcal{F}_T)$ and $m_t \in L^{\infty}(\mathcal{F}_t)$; \smallskip

\noindent - normalization: $\rho_t(0)=0$ for any $t \in \mathcal{T}$; \smallskip

\noindent - constancy: $\rho_t(m_t) = -m_t $ for any $m_t \in L^{\infty}(\mathcal{F}_t)$.\medskip

A more technical axiom is the following:\smallskip

\noindent - regularity: for any $t \in [0,T]$, $\rho_t (X 1_A + Y 1_{A^c})= 1_A \rho_t(X)+ 1_{A^c} \rho_t (Y)$ for any $A \in \mathcal{F}_t$, $X,Y \in L^{\infty}(\mathcal{F}_T)$.
\medskip

Quite recently, axioms of cash-additivity and of convexity have been discussed (see El Karoui and Ravanelli \cite{ELK-rav} and Cerreia-Vioglio et al. \cite{CMMM2}) and weakened, respectively, by:\smallskip

\noindent - cash-subadditivity: for any $t \in \mathcal{T}$,  $\rho_t (X+m_t) \geq \rho_t (X)-m_t$ for any $X \in L^{\infty}(\mathcal{F}_T)$ and $m_t \in L^{\infty}_+(\mathcal{F}_t)$; \smallskip

\noindent - quasiconvexity: $\rho_t( \alpha X + (1- \alpha) Y) \leq \esssup \{ \rho_t(X); \rho_t(Y) \}$ for any $t \in \mathcal{T}$, $X, Y \in L^{\infty}(\mathcal{F}_T)$, $\alpha \in [0,1]$. \smallskip

We postpone to the next section the discussion of cash-additivity versus cash-subadditivity.
Notice that cash-subadditivity and normalization imply that $\rho_t(m_t)=\rho_t(0+m_t) \geq \rho_t(0)-m_t=-m_t$ (and, similarly, $\rho_t(-m_t) \leq m_t$) for any $m_t \in L^{\infty}_+(\mathcal{F}_t)$.
\bigskip

From now on, we will denote by $\rho_{s,t}:L^{\infty}(\mathcal{F}_t) \to L^{\infty}(\mathcal{F}_s)$ (for $s \leq t$) and by $\rho_s= \rho_{s,T}$.

A desirable property for a dynamic risk measure is the so-called time-consistency that allows to relate the same risk measure at different times. Different notions of time-consistency exist, however, in the literature: \smallskip

\noindent - strong time-consistency (shortly, time-consistency) or recursivity:

\noindent $\rho_{s,t} \left(-\rho_{t,u} (X) \right)=\rho_{s,u} (X)$ for any $X \in L^{\infty}(\mathcal{F}_u)$ and $s,t,u$ with $0 \leq s \leq t \leq u \leq T$; \smallskip

\noindent - weak time-consistency:

\noindent if $\rho_{t,u} (X) \geq \rho_{t,u} (Y)$, then $\rho_{s,u} (X) \geq \rho_{s,u} (Y)$ for any $s \in [t,u]$; \smallskip

\noindent - weak* time-consistency:

\noindent if $\rho_{t,u} (X) = \rho_{t,u} (Y)$, then $\rho_{s,u} (X) = \rho_{s,u} (Y)$ for any $s \in [t,u]$.
\medskip

While strong time-consistency guarantees that the riskiness of a position at time $s$ can be equivalently calculated in two ways (that is, directly at time $s$ or in two steps - from time $u$ to time $t$ and then to time $s$), weak and weak* time-consistency imply that if a position is riskier than (or as risky as) another at time $t$ then the same holds at any time $s \leq t$. Further notions of time-consistency can be also found in the recent paper of Roorda and Schumacher \cite{RoSch15}.

It is well known (see F\"{o}llmer and Penner \cite{follmer-penner}, Delbaen \cite{delb-mstable} and Detlefsen and Scandolo \cite{detlef-scandolo}, among others) that for convex cash-additive risk measures the three notions above are equivalent. Moreover, for dynamic convex cash-additive risk measures Bion-Nadal \cite{bion-nadal-FS}, \cite{bion-nadal-SPA} proved that time-consistency is strongly related to the so-called cocycle property of the penalty term of the dynamic risk measure.
\medskip

The main aim of this paper is to investigate what happens in the cash-subadditive case and to provide sufficient conditions for a convex cash-subadditive risk measure to be time-consistent.
Obviously, for general risk measures weak time-consistency implies weak* time-consistency.

\section{Dual representation of Cash-Subadditive Risk Measures}\label{sec:rr}

As emphasized in El Karoui and Ravanelli \cite{ELK-rav}, assuming cash-additivity is not always reasonable for a risk measure
mainly when dealing with stochastic interest rates or ambiguity over discounting. Motivated by these arguments, the aforementioned authors proposed to replace cash-additivity with
the weaker assumption of cash-subadditivity.

\subsection{Static setting}

In the following, we recall from El Karoui and Ravanelli \cite{ELK-rav} the dual representation of convex cash-subadditive risk measures, similar to the one for convex cash-additive risk measures but in terms of subprobabilities, and we provide some additional results that will be useful in the paper.
In particular, a characterization of those convex cash-subadditive measures of risk on $L^\infty$
which can be represented by a penalty function concentrated on probability measures is given.
\smallskip

\begin{proposition}[see Theorem 4.3 in El Karoui and Ravanelli \cite{ELK-rav}]
Any convex, monotone, normalized and cash-subadditive risk measure $\rho:L^{\infty} \to \Bbb R$ can be represented as
  \begin{align}\label{eq:repElRa09}
      \rho(X)=\sup\limits_{\mu\in \mathcal{M}_{s,f}} \left\{ \bE_\mu[-X]-\bar{c}(\mu) \right\},
   \end{align}
where $\bar{c}$ is the minimal penalty function defined by
\begin{equation}\label{eq:minpen}
     \bar{c}(\mu)= \sup_{X\in L^\infty} \left\{ \bE_\mu[-X]-\rho(X)\right\}.
\end{equation}
\end{proposition}

\begin{remark}
Since any subprobability $\mu \in \mathcal{M}_{s,f}$ can be written as $\mu(\cdot)=aQ(\cdot)$ for some $a\in [0,1]$ and $Q\in \cM_{1,f}$, then representation \eqref{eq:repElRa09} (and the analogous with any penalty function) can be rewritten as follows
   \begin{align}\label{eq:rep-paper}
      \rho(X)&=\sup\limits_{a \in [0,1],Q\in \mathcal{M}_{1,f}} \left\{ a\bE_Q[-X]-\bar{c}(aQ)
 \right\}\\
 &=\sup\limits_{a \in [0,1],Q\in \mathcal{M}_{1,f}} \left\{ a\bE_Q[-X]-c(aQ) \right\}
   \end{align}
where $c$ is any penalty function.
   In the rest of the paper we will always refer to this representation, where the penalty functions
 are seen as maps on $[0,1] \times \cM_{1,f}$, instead of $\cM_{s,f}$.
\end{remark}

As recalled below, continuity from below of $\rho$ guarantees that the dual representation in \eqref{eq:repElRa09} can be done in terms of probability measures, not only of finitely additive measures.

\begin{proposition}[see Theorem 4.3 and Corollary 4.4 in \cite{ELK-rav}]
     Let $\rho: L^{\infty} \to \Bbb R$ be a convex, monotone, normalized and cash-subadditive measure of risk which is continuous from below.

Suppose that $c$ is any penalty function on $[0,1]\times \cM_{1,f}$ representing $\rho$. Then $c$ is concentrated on $[0,1]\times \cM_1$, i.e.
\begin{align*}
      c(aQ)<\infty \qquad \Longrightarrow \qquad  Q \ \textrm{is $\sigma$-additive (hence a  probability measure)}.
\end{align*}
\end{proposition}

We focus now on those risk measures which are defined on a probability space
$(\Omega, \cF,P)$. From now on $L^\infty$ will denote $L^{\infty}=L^\infty(\Omega,\cF,P)$ while $\cM_{1}(P)$ (respectively $\cM_{s}(P)$)
the set of all $\sigma$-additive probability (resp. subprobability) measures on $(\Omega,\cF)$ which are absolutely continuous with respect to $P$.

Notice that any risk measure $\rho:L^{\infty} \to \mathbb{R} \cup \{+\infty\}$ satisfying monotonicity, cash-subadditivity and normalization is finite-valued. Indeed: by monotonicity $\rho(\esssup X) \leq \rho(X) \leq \rho(\essinf X)$.  If $\essinf(X) \geq 0$, then $\rho(X) \leq \rho(\essinf X)\leq 0$ and $\rho(X) \geq \rho(\esssup X)\geq -\esssup(X) \in \mathbb R$. Otherwise, $\rho(X) \leq \rho(\essinf X)\leq -\essinf (X) \in \mathbb R$ and $\rho(X)=\rho(X-(\esssup(X)+1)+1) \geq \rho(1) \geq -1$.\medskip

The following result characterizes those convex, cash-subadditive, monotone, normalized risk measures that can be represented in terms of subprobability measures.  The proof is driven by means of Fenchel-Moreau biconjugate theorem or, in particular, by using the representation of general convex risk measures (see Frittelli and Rosazza Gianin \cite{frittelli-rg} and F\"{o}llmer and Schied \cite{follmers-book}). This approach is different from the one used in the proof of Theorem 4.3 in El Karoui and Ravanelli \cite{ELK-rav}. In that case, indeed, the aforementioned authors prove that to any cash-subadditive risk measure it corresponds a cash-additive one by adding a new dimension, hence they apply the results already known for cash-additive risk measures.

\begin{theorem}\label{thm:repstatic}
    Let $\rho : L^{\infty} \to \mathbb R$ be a convex, monotone, cash-subadditive and normalized risk measure. Then the following are equivalent:
\begin{itemize}
\item[(i)]  \label{it:cast} $\rho$ is continuous from above;
\item[(ii)]  $\rho$ is lower semi-continuous with respect to the $\sigma(L^{\infty}, L^1)$-topology;
\item[(iii)]  \label{it:repminst} $\rho$ can be represented as
\begin{equation} \label{eq:repr_rho_bar_c}
\rho(X)=\sup\limits_{a \in [0,1],Q\in \mathcal{Q}} \left\{ a\bE_Q[-X]-\bar{c}(aQ)
 \right\}
\end{equation}
where $\mathcal{Q} \subseteq \mathcal{M}_1(P)$ and where the minimal penalty function $\bar{c}:[0,1] \times \mathcal{Q} \to [0;+\infty]$ is defined by
\begin{align}\label{eq:rep-gen}
 \bar{c}(aQ)=\sup_{X\in L^\infty} \left\{  a\bE_{Q}[-X]-\rho(X)\right\}.
\end{align}
\item[(iv)]  \label{it:repst} $\rho$ can be represented as in \eqref{eq:repr_rho_bar_c} by means of some penalty function
$c: [0,1]\times \cQ \to [0; +\infty]$.
\end{itemize}
\end{theorem}

The proof of the previous result is postponed to the Appendix (see section \ref{appendix: proof 1}).

\subsection{Dynamic setting}

In the previous section we saw that any (static) convex cash-subadditive and continuous from above risk measure $\rho: L^\infty(\Omega,\cF,P) \to \bR$ can be represented as
\begin{equation} \label{eq: static_repres}
     \rho(X)=\sup_{a\in [0,1],Q  \in \cQ } \left\{ a \mathbb{E}_Q[-X]-c(aQ)\right\}, \qquad X\in L^\infty,
\end{equation}
in terms of a set $\mathcal{Q} \subseteq \cM_1(P)$ of probability measures and of a {\it penalty function} $c: [0,1] \times \cQ \to [0,+\infty]$.
We focus now on a dynamic setting where we prove a similar dual representation for dynamic convex cash-subadditive risk measures that are continuous from above. To this aim we follow a different approach from the one used in El Karoui and Ravanelli \cite{ELK-rav} for the static case where the authors associated to any cash-subadditive risk measure a cash-additive one by adding a new dimension so to be able to use classical results for cash-additive risk measures. To be more precise, we will follow an approach similar to the one used by Detlefsen and Scandolo \cite{detlef-scandolo} for cash-additive risk measures.
\medskip

Let $[0,T]$ be a time interval and $(\mathcal{F}_t)_{t \in [0,T]}$ be a filtration of $\mathcal{F}$ such that $\mathcal{F}_0=\{\emptyset; \Omega \}$ and $\mathcal{F}_T=\mathcal{F}$. Denote by $\cD$ the set of all adapted stochastic processes $(D_t)_{t\in [0,T]}$ taking values in $[0,1]$ and by $\cQ_t$ the following set of probability measures reducing to $P$ on $\mathcal{F}_t$, that is
\begin{equation*}
  \cQ_t \triangleq \left\{ Q \ \textrm{on} \ (\Omega,\cF)| \ Q\ll P \ \textrm{ and } Q=P \textrm{ on } \cF_t \right\}, \quad t\in [0,T].
\end{equation*}

Moreover, in the following we will consider a generalized notion of penalty term $c_t(DQ)$ (for any $t \in [0,T]$) that is defined for any $D \in \mathcal{D}$ and $Q \in \mathcal{Q}$, and is an $\mathcal{F}_t$-measurable non-negative random variable taking also $+\infty$ as possible value. From now on, in a dynamic setting with penalty term we will mean $c_t$ (with $t \in [0,T]$) as above.\smallskip

The following result guarantees that the dynamic version of \eqref{eq: static_repres} is a dynamic convex, monotone, cash-subadditive and normalized risk measure.

\begin{proposition} \label{prop: properties of rho}
Given a penalty term $c_t$ satisfying $\essinf _{(D,Q) \in \mathcal{D} \times \mathcal{Q}} c_t(DQ)=0$, the dynamic risk measure defined by
\begin{equation} \label{eq: form-rho}
\rho_t(X)= \esssup _{(D,Q) \in \mathcal{D} \times \mathcal{Q}} \{ D_t \bE_Q[ \left. -X \right| \mathcal{F}_t] -c_t(DQ) \}, \quad X\in L^\infty (\mathcal{F}_T), t\in [0,T],
\end{equation}
is convex, monotone, cash-subadditive and normalized and taking values in $L^{\infty}(\mathcal{F}_t)$.
\end{proposition}

\begin{proof}
Convexity, monotonicity and normalization (as well as $\rho_t (\cdot) \in L^{\infty}(\mathcal{F}_t)$) are straightforward.

Cash-subadditivity: for any $t \in [0,T]$, $m_t \in L^{\infty}_+(\mathcal{F}_t)$ and $X \in L^{\infty} (\mathcal{F}_T)$ it holds that
\begin{eqnarray*}
\rho_t(X+m_t)+m_t &=& \esssup_{(D,Q) \in \mathcal{D} \times \mathcal{Q}} \{ D_t \bE_Q[ \left. -X-m_t \right| \mathcal{F}_t] -c_t(DQ) \}+m_t \\
&=& \esssup_{(D,Q) \in \mathcal{D} \times \mathcal{Q}} \{ m_t(1- D_t)+D_t \bE_Q[ \left. -X \right| \mathcal{F}_t] -c_t(DQ) \}\\
&\geq & \esssup_{(D,Q) \in \mathcal{D} \times \mathcal{Q}} \{ D_t \bE_Q[ \left. -X \right| \mathcal{F}_t] -c_t(DQ) \}=\rho_t(X),
\end{eqnarray*}
hence the thesis.
\end{proof}

Under a continuity and a regularity condition also a converse result holds, as shown below.

\begin{proposition} \label{prop: repr-dyn-csa}
    Let $(\rho_t)_{t\in [0,T]}$ be a dynamic convex, monotone, cash-subadditive, normalized and regular risk measure with $\rho_t: L^\infty (\mathcal{F}_T) \to L^\infty(\cF_t)$. Then the following are equivalent:
\begin{enumerate}
\item[(i)] for any $t \in [0,T]$, $\rho_t$ is continuous from above;
\item[(ii)]  for any $t \in [0,T]$, $\rho_t$ can be represented as
\begin{equation} \label{eq: dynamic_repres_c}
   \rho_t(X)= {\rm ess.sup}_{(D,Q) \in \cD \times \cQ_t }\left\{ D_t\bE_Q[-X|\cF_t]- c_t(D Q)\right\}, \ X\in L^\infty, \ t\in [0,T]
\end{equation}
for some penalty term $c_t$;

\item[(iii)]  for any $t \in [0,T]$, $\rho_t$ can be represented as in \eqref{eq: dynamic_repres_c} in terms of the minimal penalty term $\bar{c_t}$, that is
\begin{equation}\label{eq:ct*}
    \bar{c_t}(DQ)\triangleq{\rm ess.sup}_{X\in L^\infty} \left\{ D_t\bE_Q[-X|\cF_t]- \rho_t(X)\right\},
\end{equation}
for $(D,Q)\in \cD \times \cQ_t$.
\end{enumerate}
\end{proposition}

The proof of the previous result is postponed to the Appendix (see section \ref{appendix: proof2}).

\section{Time-consistency for cash-subadditive risk measures}\label{sec:cocycle}

In Drapeau et al. \cite{drapeau-kupper-tangpi-rg} and in El Karoui and Ravanelli \cite{ELK-rav}, dynamic convex cash-subadditive risk measures of the following form were considered in a Brownian setting:
\begin{equation} \label{eq: cash-subadd-integr f}
\rho_{t,T}(X)= \esssup _{(\beta,q) \in \mathcal{B} \times \mathcal{Q}} \left\{  E_Q \left[\left. -B_{t,T}^{\beta} X \right| \mathcal{F}_t \right] - E_Q \left[\left. \int_t ^T B_{t,s}^{\beta} f(\beta_s, q_s) ds \right| \mathcal{F}_t \right] \right\},
\end{equation}
where $\mathcal{B}$ is a set of adapted stochastic processes $\beta$, $B_{s,t}^{\beta}= \exp\{-\int_s ^t \beta_u du \}$ is the discount factor associated to $\beta$, $\mathcal{Q}$ is a set of $d$-dimensional adapted stochastic processes $q$ corresponding to probability measures $Q$ via stochastic exponentials, i.e. $E\left[\left. \frac{dQ}{dP} \right| \mathcal{F}_t \right] =\exp\{-\frac{1}{2} \int_0 ^t \Vert q_u \Vert^2 du +\int_0 ^t q_u dW_u \}$, and $f=f(\omega,t, \beta,q):\Omega \times [0,T] \times \Bbb R \times \Bbb R^d \to \Bbb R \cup \{ + \infty\}$ is a given functional.
In particular, the previous dynamic risk measures are shown to be time-consistent and the penalty term $c_{t,T}(\beta,q) =E_Q \left[\left. \int_t ^T B_{t,s}^{\beta} f(\beta_s, q_s) ds \right| \mathcal{F}_t \right]$ satisfies the following generalized cocycle property:
\begin{equation} \label{eq: penalty-from integr f}
c_{t,T}(\beta, q) =c_{t,u}(\beta, q) + E_Q \left[\left. B_{t,u}^{\beta} c_{u,T}(\beta,q) \right| \mathcal{F}_t \right]
\end{equation}
for any $t,u$ s.t. $0 \leq t \leq u \leq T$.
\medskip

Inspired by the results above and by Bion-Nadal \cite{bion-nadal-FS}, \cite{bion-nadal-SPA} (where the classical cocycle property has been shown to be related to time-consistency of dynamic convex cash-additive risk measures), we ask whether the generalized cocycle - eventually together with other conditions - guarantees that the corresponding dynamic convex cash-subadditive risk measure is time-consistent or, better, under which conditions a general dynamic cash-subadditive risk measures is time-consistent. In this section we will give an answer to the question above by following an approach similar to the one used by Bion-Nadal \cite{bion-nadal-FS}, \cite{bion-nadal-SPA} for cash-additivite risk measures.
\medskip

As previously, let $\mathcal{D}$ be the set of all adapted stochastic processes $(D_t)_{t \in [0,T]}$ taking values in $[0,1]$ and let $\mathcal{Q}$ be a set of probability measures (absolutely continuous or equivalent to $P$).
Let now $(\rho_t)_{t \in [0,T]}$ be a dynamic risk measure of the following form:
\begin{equation} \label{eq: form-rho-hypoth}
\rho_t(X)= \esssup _{(D,Q) \in \mathcal{D} \times \mathcal{Q}} \{ D_t \bE_Q[ \left. -X \right| \mathcal{F}_t] -c_t(DQ) \}
\end{equation}
for any $t \in [0,T]$ and $X \in L^{\infty}(\mathcal{F}_t)$ or, more in general,
\begin{equation} \label{eq: form-rho-s,t}
\rho_{t,u}(Y)= \esssup _{(D,Q) \in \mathcal{D} \times \mathcal{Q}} \{ D_{t,u} \bE_Q[ \left. -Y \right| \mathcal{F}_t] -c_{t,u}(DQ) \}
\end{equation}
for any $t,u \in [0,T]$ with $t \leq u$ and $Y \in L^{\infty}(\mathcal{F}_u)$.
With $(D_{t,u})_{t \in [0,u]}$ we mean an adapted stochastic process with time horizon $u$ and taking values in $[0,1]$, with $(c_{t,u})_{t \in [0,u]}$ the generalized penalty term referring to a time horizon $u$, while with $D_t=D_{t,T}$, $c_t=c_{t,T}$ and $\rho_t = \rho _{t,T}$. $D_{t,u}$ can be interpreted as a discount factor. Notice that $D_{t,u}$ is $\mathcal{F}_t$-measurable while in Drapeau et al. \cite{drapeau-kupper-tangpi-rg} and in El Karoui and Ravanelli \cite{ELK-rav} $B_{t,u}$ was $\mathcal{F}_{u}$-measurable. \bigskip

It is worth to emphasize that assuming that $\rho _t$ is of the form \eqref{eq: form-rho-hypoth} is not too restrictive. Indeed, from Proposition \ref{prop: properties of rho} we know that $(\rho_t)_{t \in [0,T]}$ is a dynamic convex, monotone and cash-subadditive risk measure, while from Proposition \ref{prop: repr-dyn-csa} we know that, under continuity from above and regularity, any dynamic convex, monotone and cash-subadditive risk measure can be represented as in \eqref{eq: form-rho-hypoth} once $\mathcal{Q}$ is replaced by $\mathcal{Q}_t$.
\medskip

We consider now the following  assumptions:
\medskip

\begin{itemize}

\item[-] \textbf{assumptions on $\mathcal{Q}$:}

\begin{enumerate}

\item[(Qa)] $\mathcal{Q}$ is a subset of probability measures that are all equivalent to the initial $P$.

\item[(Qb)] (stability for pasting) for any $Q_1, Q_2 \in \mathcal{Q}$ and for any $s,t,u \in [0,T]$ with $s \leq t \leq u$ there exists a probability measure $Q^* \in \mathcal{Q}$ (pasting between $Q_1$ and $Q_2$) such that:
$$
\bE_{Q^*} \left[\left. X \right| \mathcal{F}_s \right]=\bE_{Q_1} \left[\left. \bE_{Q_2} \left[\left. X \right| \mathcal{F}_t \right] \right| \mathcal{F}_s \right], \ \text{ for any } X \in L^{\infty}(\mathcal{F}_u).
$$

\item[(Qc)] (stability for bifurcation) for any $Q_1, Q_2 \in \mathcal{Q}$, for any $s,t \in [0,T]$ with $s \leq t$ and for any $A \in \mathcal{F}_s$ there exists a probability measure $Q^* \in \mathcal{Q}$ such that:
$$
\bE_{Q^*} \left[\left. X \right| \mathcal{F}_s \right]= 1_A \bE_{Q_1} \left[\left. X \right| \mathcal{F}_s \right]+ 1_{A^c} \bE_{Q_2} \left[\left. X \right| \mathcal{F}_s \right], \ \text{ for any } X \in L^{\infty}(\mathcal{F}_t).
$$
\end{enumerate}

\item[-] \textbf{assumption on $\mathcal{D}$:}

\begin{enumerate}

\item[(Da)] for any $D^1, D^2 \in \mathcal{D}$, for any $s,t \in [0,T]$ with $s \leq t$ and for any $A \in \mathcal{F}_s$ there exists $D^* \in \mathcal{D}$ such that:
$$
D^* _{s,t}= 1_A D^1 _{s,t}+ 1_{A^c} D^2 _{s,t}.
$$

\end{enumerate}

\item[-] \textbf{assumption on $\mathcal{Q}$ and $\mathcal{D}$ jointly:}

\begin{enumerate}
\item[(QDa)] (stability for joint pasting) for any $Q_1, Q_2 \in \mathcal{Q}$, $D^1, D^2 \in \mathcal{D}$ and for any $s,t,u \in [0,T]$ with $s \leq t \leq u$ there exist $Q^* \in \mathcal{Q}$ and $D^* \in \mathcal{D}$ such that:
$$
D^* _{s,u} \bE_{Q^*} \left[\left. X \right| \mathcal{F}_s \right]= D^1 _{s,t} \bE_{Q_1} \left[\left. D^2 _{t,u}  \bE_{Q_2} \left[\left. X \right| \mathcal{F}_t \right] \right| \mathcal{F}_s \right], \ \text{ for any } X \in L^{\infty}(\mathcal{F}_u).
$$
Moreover: when both $Q_1=Q_2=Q$ and $D^1=D^2=D$ are satisfied, it holds that $D^*=D$ and $Q^*=Q$.
\end{enumerate}

\item[-] \textbf{assumptions on the penalty term $c$:}

\begin{enumerate}

\item[(Ca)] $c_{s,t}(DQ)$ is a $\mathcal{F}_s$-measurable non-negative random variable taking, eventually, also $+ \infty$ as possible value and defined for any $D \in \mathcal{D}$ and $Q \in \mathcal{Q}$.

\item[(Cb)] (generalized locality) for any $Q_1, Q_2 \in \mathcal{Q}$, $D^1, D^2 \in \mathcal{D}$, $s,t \in [0,T]$ with $s \leq t $ and for any $A \in \mathcal{F}_s$ it holds that:
\begin{eqnarray*}
&&\mbox{if } 1_A \bE_{Q_1} \left[\left. X \right| \mathcal{F}_s \right]= 1_A \bE_{Q_2} \left[\left. X \right| \mathcal{F}_s \right] \text{ for any } X \in L^{\infty}(\mathcal{F}_t) \, \text{ and } 1_A D^1_{s,t}=1_A D^2 _{s,t} \\
&& \Rightarrow 1_A c_{s,t} (D^1 Q^1)=1_A c_{s,t} (D^2 Q^2).
\end{eqnarray*}

\item[(Cc)] (generalized cocycle) for any $Q \in \mathcal{Q}$, $D \in \mathcal{D}$, for any $s,t,u \in [0,T]$ with $s \leq t \leq u$:
\begin{equation} \label{eq: cocycle}
c_{s,u}(DQ)=c_{s,t}(DQ)+ \bE_Q \left[\left. D_{s,t} c_{t,u} (DQ) \right| \mathcal{F}_s \right].
\end{equation}

\end{enumerate}
\end{itemize}
\smallskip

Assumptions (Qa), (Qb) and (Qc) on $\mathcal{Q}$ are the the same as in Bion-Nadal \cite{bion-nadal-FS}, where (Qc) is called stability for bifurcation. Assumptions (Cb) and (Cc) on $c$ reduce to the classical ones (see Bion-Nadal \cite{bion-nadal-FS}, \cite{bion-nadal-SPA}) when $\mathcal{D}=\{D: D \equiv 1 \}$ (corresponding to the cash-additive case). Assumptions on $\mathcal{D}$ as well as the assumption (QDa) on $\mathcal{D}$ and $\mathcal{Q}$ jointly are new. Notice that this last assumption reduces to stability for pasting on $\mathcal{Q}$ (Qb) when $\mathcal{D}=\{D: D \equiv 1 \}$.

\begin{theorem} \label{thm: time-cons}
If $(\rho_{s,t})_{0 \leq s \leq t \leq T}$ is defined as in \eqref{eq: form-rho-s,t}, with $\mathcal{D}$, $\mathcal{Q}$ and $c$ satisfying all the assumptions above, then it is (strongly) time-consistent.
\end{theorem}

The proof of the previous result can be found in the appendix (see Section \ref{appendix: proof3}).

\section{Different notions of time-consistency and their relations}\label{sec:properties}

under the assumption of cash-additivity, a key question concerning dynamic convex risk measures is whether there is any relation between the same dynamic risk measure at different times. Although this question has lead to the introduction of several notions of time-consistency, for dynamic convex cash-additive risk measures it is well known (see Proposition 1.16 in \cite{AcPe11}, among others) that the three main notions of time-consistency (strong, weak and weak*) are all equivalent.
In the following, we will investigate what happens for cash-subadditive convex risk measures
and, in particular, whether a similar result still holds (once cash-additivity is weakened by cash-subadditivity).

We will see that in the cash-subadditive case the equivalence of strong, weak and weak* time-consistency is no longer valid while only some implications remain true.
The following result emphasizes the link between strong, weak and weak* time-consistency when only cash-subadditivity holds. On the one hand, this result can be seen as an
extension of Proposition 1.16 in \cite{AcPe11} to the cash-subadditive case; on the other hand, it is weaker than the aforementioned result.

\begin{proposition} \label{prop: link_time-cons2}
Let $(\rho _{t,T})_{t \in [0,T]}$ be a dynamic risk measure.

a) If it satisfies monotonicity, then time-consistency implies weak time-consistency (hence also weak* time-consistency).

b) If it satisfies constancy, then weak* time-consistency implies time-consistency.

c) If it satisfies cash-subadditivity and $\rho_{t}(0)=0$ (normalization), then weak time-consistency implies that
\begin{eqnarray*}
\rho_{s} (X) \leq \rho_{s} \left(-\rho_{t} (X) \right) \\
(\text{resp. } \rho_{s} (X) \geq \rho_{s} \left(-\rho_t (X) \right) )
\end{eqnarray*}
for any  $s,t$ with $0 \leq s \leq t \leq T $ and any $X \in L^{\infty}(\mathcal{F}_T)$ s.t. $\rho_{t} (X) \leq 0$ (resp. $\rho_{t} (X) \geq 0$).
\end{proposition}

\begin{proof}
a) can be proved exactly as in F\"{o}llmer and Penner \cite{follmer-penner} where cash-additivity is not needed.\smallskip

b) The proof can be done similarly as in F\"{o}llmer and Penner \cite{follmer-penner}.

By constancy, indeed, $\rho_{t}(X)= \rho_{t} (-\rho_{t} (X))$ holds for any $X \in L^{\infty}(\mathcal{F}_t)$ and $t \in [0,T]$.
Weak* time-consistency thus implies that $\rho_{s}(X)= \rho_{s} (-\rho_{t} (X))$ for any $s \in [0,t]$, hence time-consistency.\smallskip

c) Take any $s,t$ with $0 \leq s \leq t \leq T $ and any $X \in L^{\infty}(\mathcal{F}_T)$ s.t. $\rho_{t} (X) \leq 0$.
By $\rho_t(X) \leq 0$ and cash-subadditivity it follows that $\rho_{t} (-\rho_{t} (X)) \geq \rho_{t}(X)$. So, by weak time-consistency, $\rho_{s} (-\rho_{t} (X)) \geq \rho_{s}(X)$.

The case where $\rho_{t} (X) \leq 0$ can be checked similarly.
\end{proof}
\medskip

The above result provides the following implications for monotone risk measures:
\begin{itemize}
\item strong time-consistency $\Rightarrow$ weak time-consistency ($\Rightarrow$ weak* time-consistency)
\item \textit{under constancy:} strong time-consistency $\Leftarrow$ weak* time consistency
\end{itemize}

The following example emphasizes that weak time-consistency does not guarantee strong time-consistency in general when cash-additivity is replaced by cash-subadditivity,.

\begin{example} \rm
     Consider the following dynamic risk measure:
		\begin{align*}
		    \rho_t(X) = \frac{1}{\gamma_t} \bE\left[(-X)^+|\cF_t\right],
				\qquad X \in L^\infty, \ t\geq 0,
		\end{align*}
		with $\gamma_t \in L^\infty(\mathcal{F}_t)$ satisfying $\gamma_t > 1$.
		Such kind of dynamic risk measure generalizes to the dynamic setting the notion of put premium risk measure introduced in
		El Karoui and Ravanelli \cite[pag. 569]{ELK-rav}. As in Corollary 3.4 of \cite{ELK-rav}
		it can be proved that $(\rho_t)_{t\geq 0}$ is a dynamic cash-subadditive risk measure. Moreover, $\rho_t$ is convex, positively homogeneous and monotone
		for any $t\geq 0$.
		
		We prove now that $(\rho_t)_{t \geq 0}$ is weakly time-consistent but \textit{not} strongly time-consistent.
		
		\noindent {\it Weak time-consistency of $(\rho_t)_{t\geq 0}$}.
		Suppose that for a given $t\geq 0$ and $X,Y\in L^\infty$ we have $\rho_{t}(X) \leq \rho_{t}(Y)$. This is equivalent to say that
		\begin{align}\label{eq:tc}
		     \bE\left[(-X)^+|\cF_{t}\right] \leq \bE\left[(-Y)^+|\cF_{t}\right].
		\end{align}
		Now consider $\rho_s(X)$ and $\rho_s(Y)$ for any $s \in [0,t]$. By definition of $\rho_s$ and by inequality \eqref{eq:tc} we get
		\begin{align*}
		     \rho_s(X) &= \gamma_s^{-1}\bE\left[(-X)^+|\cF_{s}\right] =
				 \gamma_s^{-1}\bE\left[\bE\left[(-X)^+|\cF_{t}\right]|\cF_s\right]\\
				 & \leq \gamma_s^{-1}\bE\left[\bE\left[(-Y)^+|\cF_{t}\right]|\cF_s\right] =\gamma_s^{-1}\bE\left[(-Y)^+|\cF_{s}\right]	= \rho_s(Y).			
		\end{align*}
		Hence $(\rho_t)_{t\geq 0}$ is weakly time-consistent.
		\smallskip
		
		\noindent {\it Non strong time consistency}. We are going to prove the following strict inequality:
		$$
		    \rho_s (-\rho_{t}(X)) < \rho_s (X), \qquad \mbox{ for some } X \mbox{ and } 0 \leq s \leq t,
		$$
		which clearly implies that $(\rho_t)_{t\geq 0}$ cannot be strong time-consistent.
		
        For $s \in [0,t]$, we have
		\begin{align*}
		     \rho_s (-\rho_{t}(X))&= \rho_s\left(
				-\gamma_{t}^{-1}\bE\left[(-X)^+|\cF_{t}\right]\right)\\
				&= \gamma_{s}^{-1}\bE\left[\gamma_{t}^{-1}\bE\left[(-X)^+|\cF_{t}\right]|\cF_s\right].		
		\end{align*}
		Take now any $X$ with $(-X)^+ >0$ $P$-a.s. (e.g. $X=-m$ for $m \in \mathbb R$, $m>0$). Since
		$\gamma_{t}$ is $\cF_{t}$-measurable, we conclude that
		\begin{align*}
		     \rho_s (-\rho_{t}(X))&=
				\gamma_{s}^{-1}\bE\left[\bE\left[\gamma_{t}^{-1}(-X)^+|\cF_{t}\right]|\cF_s\right]\\
				&=\gamma_{s}^{-1}\bE\left[\gamma_{t}^{-1}(-X)^+|\cF_{s}\right]\\
&< \gamma_{s}^{-1}\bE\left[(-X)^+|\cF_{s}\right] =\rho_s (X),		
		\end{align*}
hence strong time-consistency fails to be satisfied by $(\rho_t)_t$.
\end{example}

\section{Concluding remarks}\label{conclusion}

The main results obtained in the paper can be summarized as follows: first, we provide a dual representation of dynamic convex cash-subadditive risk measures by means of a penalty term and of discount factors, result obtained by following a different approach from the one used by El Karoui and Ravanelli \cite{ELK-rav} for static risk measures; second, suitable conditions on the penalty term, on the discount factors and on the set of probability measures have been proved to be sufficient for a dynamic convex cash-subadditive risk measure to be strongly time-consistent; finally, we investigate which relations between the notions of strong, weak and weak* time-consistency hold true in the cash-subadditive case.

\section{Appendix: Proofs} \label{appendix}

\subsection{Proof of Theorem \ref{thm:repstatic}}\label{appendix: proof 1}

$(i) \Leftrightarrow (ii)$ is due to Proposition 2.5 of Frittelli and Maggis \cite{fritt-mag-dyn}.\smallskip

$(ii) \Rightarrow (iii)$.
    Suppose that $\rho$ is lower semi-continuous. Since $\rho$ is finite-valued, convex, monotone and lower semi-continuous, by Frittelli and Rosazza Gianin \cite[Theorem 5]{frittelli-rg} and F\"{o}llmer and Schied \cite[Theorem A.61]{follmers-book} it follows that $\rho$ is representable as follows
		\begin{align}\label{eq:repFM}
		    \rho(X)= \sup_{X'\in\cP} \left\{ X^\prime(-X) - \rho^*(X^\prime)\right\},
		\end{align}
that is in terms of the Fenchel-Moreau conjugate $\rho^*$ of $\rho$ and of a non-empty convex set
		$\cP \subseteq L^1_+$. We may suppose without loss of generality that all the elements in $\mathcal{P}$ satisfy $\rho^*(X') <+\infty$.

		It is sufficient to prove that
		\begin{equation}\label{eq:Psubset}
		\cP \subseteq \cW \triangleq \left\{ X^\prime \in L^1_+: X^\prime(1) \leq 1\right\}.
		\end{equation}
		If \eqref{eq:Psubset} holds, indeed, we can identify any $X^\prime \in \cW$ with a subprobability measure $\mu \in \mathcal{P}^\prime \subseteq \cM_s(P)$ by
setting $X^\prime=\frac{{\rm d}\mu}{{\rm d}P}$. By using indifferently the following notations
$X^\prime(X)=\bE[X^\prime (-X)]= \bE_\mu[-X]$ for $X^\prime \in \cW$,
\begin{align*}%\label{eq:repFM}
		    \rho(X)&= \sup_{X'\in\cP} \left\{ X^\prime(-X) - \rho^*(X^\prime)\right\}\\
				&= \sup_{\mu \in \cP^\prime} \left\{ \bE_\mu[-X
				] - \bar{c}(\mu)\right\}\\
				&= \sup_{(a,Q)\in [0,1]\times\cQ} \left\{ a\bE_Q[-X
				] - \bar{c}(aQ)\right\},
		\end{align*}
		for some $\mathcal{Q} \subseteq \mathcal{M}_1(P)$ since, by definition, the minimal penalty function
		$\bar{c}$ can be identified with $\rho^*$.\smallskip

		Let us, therefore, prove \eqref{eq:Psubset}. By cash-subadditivity and normalization we have
		$\rho(-m) \leq m$ for any $m \in \bR$ with $m\geq 0$.
		Hence,
		\begin{align*}
		   \rho(-m)= \sup_{X^\prime \in \cP} \left\{ m X^\prime(1)- \rho^*(X^\prime)\right\} \leq m, \quad \forall \ m \in \bR, m\geq 0.
		\end{align*}
		This implies that
        $$
        m X^\prime(1) - \rho^*(X^\prime) \leq m, \quad \forall m\geq 0, \forall X^\prime\in \cP,
		$$
		or, equivalently,
		\begin{align}\label{eq:mu<1}
		   m(X^\prime(1)-1)\leq \rho^*(X^\prime), \quad \forall m\geq 0, \forall X^\prime\in \cP.
		\end{align}
		Suppose now that $X^\prime(1)>1$. Since inequality \eqref{eq:mu<1} holds for any $m\geq 0$, then it would imply $\rho^*(X^\prime) =+\infty$.
		Hence the thesis.\smallskip

Implication $(iii) \Rightarrow (iv)$ is obvious, while $(iv) \Rightarrow (i)$ can be proved similarly to Lemma 4.20 in F\"{o}llmer and Schied \cite{follmers-book}.

\subsection{Proof of Proposition \ref{prop: repr-dyn-csa}} \label{appendix: proof2}

The present proof extends the one of Detlefsen and Scandolo \cite[Theorem 1]{detlef-scandolo} to the case of dynamic convex cash-subadditive risk measures.\medskip

The implication \textit{(iii) $\Rightarrow$ (ii)} is immediate.\smallskip

\textit{(ii) $\Rightarrow$ (i)} can be proved similarly to Detlefsen and Scandolo \cite{detlef-scandolo}. We include such a proof for completeness.

Suppose that $\rho_t$ can be represented as in \eqref{eq: dynamic_repres_c} by means of a penalty term $c_t$ and assume that $X_n \searrow X$ $P$-a.s. By Theorem of Monotone Convergence it follows that
\begin{align*}
    D_t\bE_Q[-X_n|\cF_t] -c_t(D Q) \nearrow  D_t\bE_Q[-X|\cF_t] -c_t(D Q)
\end{align*}
for every $Q\in \cQ_t$. Hence
\begin{align*}
     \rho_t(X)&={\rm ess.sup}_{(D,Q)\in \cD \times \cQ_t} \left\{ \lim_{n\to +\infty}  \left\{
  D_t\bE_Q[-X_n|\cF_t] -c_t(DQ)\right\} \right\}\\
&\leq  \liminf_{n\to +\infty} \left[{\rm ess.sup}_{(D,Q)\in \cD \times \cQ_t} \left\{
  D_t\bE_Q[-X_n|\cF_t] -c_t(DQ)\right\} \right]\\
&=\liminf_{n\to +\infty} \rho_t(X_n) \leq \rho_t(X),
\end{align*}
where the last inequality is due to monotonicity of $\rho_t$.
\smallskip

\textit{(i) $\Rightarrow$ (iii).} Suppose that, for any $t\in [0,T]$, $\rho_t$ is continuous from above. We have to prove that
$$
\rho_t(X)= {\rm ess.sup}_{(D,Q) \in \cD \times \cQ_t} \left\{ D_t\bE_Q[-X|\cF_t] -\bar{c}_t(D Q)\right\}.
$$
Since the inequality $\geq $ follows immediately from the definition of $\bar{c}_t$, it remains to show the reverse inequality. To this aim, it is sufficient to prove that, for any $t\in [0,T]$, it holds
    \begin{align*}
  \bE[ \rho_t(X)] \leq\bE \left[ {\rm ess.sup}_{(D,Q) \in \cD \times \cQ_t} \left\{ D_t\bE_Q[-X|\cF_t]-  \bar{c}_t (D Q) \right\} \right].
\end{align*}
In that case, indeed, the random variable
$$
Y \triangleq \rho_t(X) - {\rm ess.sup}_{(D,Q) \in \cD \times \cQ_t} \left\{ D_t\bE[-X|\cF_t]-  \bar{c}_t (D Q) \right\} \geq 0, \quad P\textrm{-a.s.},
$$
satisfies $\bE[Y]\leq 0$, implying that $Y=0$, $P$-a.s., hence the thesis.

We proceed now in successive steps.\medskip

\textit{Step 1: definition and properties of $\rho_{0,t}$.}

Let now the map $\rho_{0,t}: L^\infty \to \bR$ be defined as $\rho_{0,t}(X)\triangleq \mathbb{E}[\rho_t(X)]$
for $X\in L^\infty$. It is immediate to check that $\rho_{0,t}$ is a static convex, monotone, cash-subadditive risk measure. Furthermore, $\rho_{0,t}$ is continuous from above. Taking indeed any sequence $X_n \searrow X$, $P$-a.s., by monotone convergence and continuity from above of $\rho_t$ it holds that
$$
   \rho_{0,t}(X_n)=\bE[\rho_t(X_n)] \nearrow \bE [\rho_t(X)]=\rho_{0,t}(X).
$$
By the arguments above and by Theorem \ref{thm:repstatic} it follows that $\rho_{0,t}$ can be represented as
\begin{eqnarray}
     \rho_{0,t}(X) &=& \sup_{(a,Q) \in [0,1]\times\cQ}\left\{ a\bE_Q[-X]-\bar{c}_{0,t}(aQ)\right\} \notag\\
     &=& \sup_{\mu \in \mathcal{S} }\left\{ \bE_\mu[-X]-\bar{c}_{0,t} (\mu)\right\} \label{eq:rho_0,t}
\end{eqnarray}
where $\cQ \subseteq \cM_{1}(P)$, $\mathcal{S} \subseteq \cM_{s}(P)$ and
\begin{eqnarray*}
          \bar{c}_{0,t} (aQ) &=& \sup_{X\in L^\infty} \left\{ a\bE_{Q}[-X]-\rho_{0,t}(X)\right\},  \quad { \rm \ for \ any \ } Q \in \cQ \\
          \bar{c}_{0,t}(\mu) &=& \sup_{X\in L^\infty} \left\{ \bE_{\mu}[-X]-\rho_{0,t}(X)\right\}, \qquad {\rm \ for\ any\ } \mu \in \mathcal{S}.
\end{eqnarray*}
We need to prove that $\rho_{0,t}$ can be written also in the following way:
\begin{equation} \label{eq: rho-0,t-2}
    \rho_{0,t}(X) = \sup_{(D,Q) \in \mathcal{D}\times\cQ_t}\left\{ \mathbb{E}_{P}[D_t\bE_{Q_t}[-X|\cF_t]]-\bar{c}_{0,t}(D Q)\right\}.
\end{equation}
\smallskip

\textit{Step 2: $\mu$ can be decomposed as $\mu=D_t \tilde{Q}$ for some $D_t \in \mathcal{D}$ and $\tilde{Q}\in \mathcal{Q}_t$.}

We proceed by proving that if $\mu=a Q$ (with $Q\in \cM_1(P)$, $Q \ll P$ and $a \in [0,1]$) satisfies $\bar{c}_{0,t}(\mu) < + \infty$, then
there exist $D_t \in \cD$ and $\tilde{Q} \in \cQ_t$  (hence $0 \leq D_t \leq 1$ while
$\tilde{Q}= P $ on $\cF_t$)
satisfying $\mu = D_t \tilde{Q}$.

Denote by $Z_T=\frac{dQ}{dP}$, by $Z_t \triangleq \bE_P \left[\left. \frac{dQ}{dP} \right| \mathcal{F}_t\right]$ and by $N_{0,t}$
the set $N_{0,t} \triangleq \left\{\omega \in \Omega: \, Z_t (\omega)=0 \right\}$.
Clearly, $N_{0,t}\in \cF_t$. Moreover,
we notice that ${\bf 1}_{N_{0,t}}Z_T \equiv 0$ $P$-a.s..
Indeed,
\begin{align*}
\mathbb{E}_P [{\bf 1}_{N_{0,t}} Z_T ]& =
\mathbb{E}_P[\mathbb{E}_P[{\bf 1}_{N_{0,t}} Z_T|\cF_t]] \\
&=\mathbb{E}_P[{\bf 1}_{N_{0,t}}\mathbb{E}_P[ Z_T|\cF_t]]\\
&= \mathbb{E}_P[{\bf 1}_{N_{0,t}} Z_t]=0.
\end{align*}
The argument above and
${\bf 1}_{N_{0,t}}Z_T\geq 0$, $P$-a.s., imply that
${\bf 1}_{N_{0,t}}Z_T\equiv 0$, $P$-a.s..

Set now
\begin{align*}
\tilde{Q} (B)&\triangleq
   \int_{B \setminus N_{0,t}} Z_t^{-1} (\omega){\rm d}Q( \omega)
	 + \int_{B \cap N_{0,t}} {\rm d} P( \omega)\qquad \textrm{for any $B\in \cF_T$}
\\ D_t&\triangleq aZ_t
\end{align*}
Hence $D_t \geq 0$, $P$-a.s., and $D_t=0$ on $N_{0,t}$. Furthermore, it can be checked that $\tilde{Q}$ is a probability measure.
We verify now that $\tilde{Q}=P$ on $\cF_t$. For any $B\in \cF_t$, indeed, we have
\begin{equation}\label{eq:N0t}
	\tilde{Q}(B) = \int_{B \setminus N_{0,t}}  Z_t^{-1}(\omega)Z_T(\omega){\rm d}P(\omega)+\int_{B \cap N_{0,t}} {\rm d}P(\omega).
\end{equation}
Since $B \setminus N_{0,t} \in \cF_t$, using the definition of conditional expectation, we obtain
\begin{eqnarray*}
    &&\int_{B \setminus N_{0,t}}  Z_t^{-1}(\omega)Z_T(\omega){\rm d}P(\omega) = \int_\Omega {\bf 1}_{B \setminus N_{0,t}} (\omega) Z_t^{-1}(\omega)Z_T(\omega){\rm d}P(\omega)\\
		&&=\mathbb{E}_P[ \mathbb{E}_P[{\bf 1}_{B \setminus N_{0,t}}Z_t^{-1} Z_T|\cF_t]]=
		\mathbb{E}_P[{\bf 1}_{B \setminus N_{0,t}}  Z_t^{-1}Z_t]=P(B\setminus N_{0,t}).		
\end{eqnarray*}
Using the last equality, \eqref{eq:N0t} becomes $\tilde{Q}(B)
		=P(B\setminus N_{0,t})+P(B \cap N_{0,t})=P(B)$. Hence $\tilde{Q}=P$ on $\mathcal{F}_t$.

We prove now that $\mu=D_t \tilde{Q}$ on $\cF_T$.
Indeed, for any $C\in \cF_T$
\begin{align*}
    \mu(C) &= \int_{C  \setminus N_{0,t}} a{\rm d}Q(\omega) + \int_{C  \cap N_{0,t}} a{\rm d}Q(\omega) \\
		&= \int_{C  \setminus N_{0,t}} a Z_t(\omega){\rm d}\tilde{Q}(\omega) +  \int_{C  \cap N_{0,t}} a Z_T(\omega){\rm d}P(\omega) \\
        &= \int_{C  \setminus N_{0,t}} D_t(\omega){\rm d}\tilde{Q}(\omega)+  \int_{C  \cap N_{0,t}} a Z_T(\omega){\rm d}P(\omega) \\
        &=\int_{C  } D_t(\omega){\rm d}\tilde{Q}(\omega)
		\end{align*}
since $\int_{C  \cap N_{0,t}} a Z_T(\omega){\rm d}P(\omega)=0=\int_{C  \cap N_{0,t}} D_t(\omega){\rm d}\tilde{Q}(\omega)$. Hence $\mu=D_t\tilde{Q}$.

We already know that $D_t \geq 0$, $P$-a.s.. It remains to verify that $D_t \leq 1$, $P$-a.s.. Suppose now by contradiction that $P(D_t>1)>0$ and set $A=\{ \omega: D_t (\omega)>1 \}$. Since
$D_t$ is $\cF_t$-measurable,  we have that $A\in \cF_t$ and
\begin{align*}
   \ P(A) < \int_A D_t(\omega) {\rm d}P(\omega)
	=& \ \int_{A \setminus N_{0,t}} a Z_t(\omega) {\rm d}\tilde{Q}(\omega) +\int_{A \cap N_{0,t}} a Z_t(\omega) {\rm d}P(\omega)  \\
	 =&\ \int_{A \setminus N_{0,t}} a  Z_t(\omega) Z_t^{-1}(\omega) {\rm d}Q(\omega)  +\int_{A \cap N_{0,t}} a  {\rm d}Q(\omega) \\
	 =&\ aQ(A\setminus N_{0,t}) +aQ(A\cap N_{0,t}) =\ \mu(A).
\end{align*}
By cash-subadditivity and regularity of $\rho_t$, it follows therefore that
\begin{align*}
\bar{c}_{0,t} (\mu) &= \sup_{X \in L^{\infty}} \{ \mathbb{E}_{\mu} [-X] - \rho_{0,t}(X) \} \\
&\geq \sup_{\lambda >0} \{ \mathbb{E}_{\mu} [\lambda 1_A] - \rho_{0,t}(-\lambda 1_A) \} \\
&= \sup_{\lambda >0} \{ \lambda \mu(A) - \mathbb{E}[\rho_{t}(-\lambda 1_A)] \} \\
&= \sup_{\lambda >0} \{ \lambda \mu(A) - \mathbb{E}[1_A\rho_{t}(-\lambda) ] \} \\
&\geq \sup_{\lambda >0} \{ \lambda (\mu(A) - P(A)) \}=+\infty
\end{align*}
that is a contradiction with the assumption $\bar{c}_{0,t} (\mu) <+\infty$.
Hence $0 \leq D_t \leq 1$, $P$-a.s.

By all the arguments above we deduce that
\begin{equation*}
    \mathbb{E}_\mu[X] =\mathbb{E}_{\tilde{Q}}[D_t X]=\mathbb{E}_{\tilde{Q}}[\mathbb{E}_{\tilde{Q}}[ D_t  X|\cF_t]] =  \mathbb{E}_{P}[D_t\mathbb{E}_{\tilde{Q}}[ X|\cF_t]],
\end{equation*}
hence
\begin{equation*}
\bar{c}_{0,t}(DQ) = \sup_{X \in L^{\infty}} \{\mathbb{E}_{P}[D_t\mathbb{E}_{\tilde{Q}}[ X|\cF_t]] - \rho _{0,t} (X) \}.
\end{equation*}
\smallskip

\textit{Step 3: $\bE[\bar{c}_{t}(DQ)]=\bar{c}_{0,t}(DQ)$ for any $D \in \mathcal{D}$, $Q \in \cQ_t$.}

We prove now that $\bE[\bar{c}_{t}(DQ)]=\bar{c}_{0,t}(DQ)$ for any $D \in \mathcal{D}$ and $Q \in \cQ_t$.
Fix $(D,Q) \in \cD \times \cQ_t$ arbitrarily and set
\begin{align*}
   \cC_{D,Q}\triangleq    \left\{ D_t \bE_{Q}[-X|\cF_t]-\rho_t(X)| X\in L^\infty\right\}.
\end{align*}
We claim that $\cC_{D,Q}$ is upward directed, that is for any $X,Y\in L^\infty$ there exists
$\bar{X} \in L^{\infty}$ such that
$$
D_t\bE_{Q}[-\bar{X}|\cF_t]-\rho_t(\bar{X})= \max\left\{D_t\bE_{Q}[-X|\cF_t]-\rho_t(X);D_t\bE_{Q}[-Y|\cF_t]-\rho_t(Y)\right\}
$$
(hence belonging to $\cC_{D,Q}$).

Indeed, fix $X,Y \in L^\infty$ and set $Z\triangleq X {\bf 1}_A + Y {\bf 1}_{A^c} \in L^\infty$, where
$$
     A \triangleq \left\{ D_t\bE_{Q}[-X|\cF_t]-\rho_t(X) \geq D_t\bE_{Q}[-Y|\cF_t]-\rho_t(Y)\right\}.
$$
Obviously, $A\in \cF_t$.
By the regularity of $\rho_t$,
$$
     \rho_t(Z)= \rho_t(X {\bf 1}_A + Y {\bf 1}_{A^c} ) =  {\bf 1}_A \rho_t(X ) +{\bf 1}_{A^c} \rho(Y).
$$
Hence
\begin{eqnarray*}
     &&D_t\bE_{Q}[-Z|\cF_t]-\rho_t(Z) \\
     &&= D_t\bE_{Q}[-X {\bf 1}_A - Y {\bf 1}_{A^c}|\cF_t]-\rho_t(X {\bf 1}_A + Y {\bf 1}_{A^c})  \\
  &&=  (D_t\bE_{Q}[-X |\cF_t]-\rho_t(X)){\bf 1}_A + (D_t\bE_{Q}[-Y |\cF_t]-\rho_t(Y)){\bf 1}_{A^c}\\
&&=  \max\left\{D_t\bE_{Q}[-X |\cF_t]-\rho_t(X)); D_t\bE_{Q}[-Y |\cF_t]-\rho_t(Y))\right\}.
\end{eqnarray*}
Therefore, $\cC_{D,Q}$ is upward directed.
By F\"{o}llmer and Schied \cite[Theorem A.32]{follmers-book}, it follows that
$$
  \bE\left[  {\rm ess. sup} \ \cC_{D,Q}\right]=   {\rm ess.sup}_{X \in L^\infty}
 \bE \left[ D_t\bE_{Q}[-X |\cF_t]-\rho_t(X)\right]
$$
Hence, for any $(D,Q)\in \cD\times \cQ_t$,
\begin{align*}
    \bE[\bar{c}_t(D Q)] &= \bE\left[{\rm ess. sup}_{ X \in L^\infty} \left\{D_t\bE_{Q}[-X |\cF_t]-\rho_t(X)\right\} \right] \\
& =  {\rm ess. sup}_{X \in L^\infty} \bE \left\{ D_t\bE_{Q}[-X |\cF_t]-\bE[\rho_t(X)]\right\} \\
&={\rm ess. sup}_{X \in L^\infty} \left\{\bE[D_t\bE_{Q}[-X |\cF_t]]-\rho_{0,t}(X)\right\}\\
&=\bar{c}_{0,t} (D Q ).
\end{align*}
\smallskip

\textit{Step 4: final arguments.}
Finally, by \eqref{eq:rho_0,t} we obtain
\begin{align*}
 \rho_{0,t}(X) &=  {\rm sup}_{(D,Q)\in \cD\times \cQ_t} \left\{\bE[D_t\bE_{Q}[-X|\cF_t ]]-\bE[\bar{c}_{t} (D Q)]\right\} \\
&\leq  \bE \left[   {\rm sup}_{(D,Q)\in \cD\times \cQ_t} \left\{D_t\bE_{Q}[-X |\cF_t]]-\bar{c}_{t} (D Q)]\right\}  \right]
\end{align*}
that completes the proof.

\subsection{Proof of Theorem \ref{thm: time-cons}} \label{appendix: proof3}

The present proof is in line with the one of Theorem 4.4 of Bion-Nadal \cite{bion-nadal-FS} for dynamic convex and cash-additive risk measures.

Let $s,t \in [0,T]$ (with $s \leq t$) and $X \in L^{\infty}(\mathcal{F}_t)$ be fixed arbitrarily and set
$$
\mathcal{C}_X \triangleq \left\{D_{s,t} \bE_Q \left[\left. -X \right| \mathcal{F}_s \right]-c_{s,t}(DQ) | D \in \mathcal{D}, Q \in \mathcal{Q} \right\}.
$$
Let us prove that $\mathcal{C}_X$ is upward directed, that is: for any $D^1,D^2 \in \mathcal{D}$, $Q_1,Q_2 \in \mathcal{Q}$ there exist $\bar{D} \in \mathcal{D}$, $\bar{Q} \in \mathcal{Q}$ such that
$$
\max_{i=1,2} \left\{D^i_{s,t} \bE_{Q_i} \left[\left. -X \right| \mathcal{F}_s \right]-c_{s,t}(D^i Q_i) \right\}=\bar{D}_{s,t} \bE_{\bar{Q}} \left[\left. -X \right| \mathcal{F}_s \right]-c_{s,t}(\bar{D}\bar{Q})
$$
(hence belonging to $\mathcal{C}_X$).

Let $D^1,D^2 \in \mathcal{D}$ and $Q_1,Q_2 \in \mathcal{Q}$ and set
\begin{equation} \label{eq: def_A_for-upward}
A \triangleq \left\{D^1_{s,t} \bE_{Q_1} \left[\left. -X \right| \mathcal{F}_s \right]-c_{s,t}(D^1 Q_1) > D^2_{s,t} \bE_{Q_2} \left[\left. -X \right| \mathcal{F}_s \right]-c_{s,t}(D^2 Q_2) \right\}.
\end{equation}
Obviously, $A \in \mathcal{F}_s$.
By stability of $\mathcal{D}$ and $\mathcal{Q}$, there exist $\bar{D} \in \mathcal{D}$ and $\bar{Q} \in \mathcal{Q}$ such that
\begin{eqnarray*}
\bE_{\bar{Q}} \left[\left. Y \right| \mathcal{F}_s \right] &=& 1_A \bE_{Q_1} \left[\left. Y \right| \mathcal{F}_s \right] + 1_{A^c} \bE_{Q_2} \left[\left. Y \right| \mathcal{F}_s \right], \quad \forall Y \in L^{\infty}(\mathcal{F}_t) \\
\bar{D}_{s,t} &=& 1_A D^1 _{s,t} + 1_{A^c} D^2 _{s,t}.
\end{eqnarray*}
By the arguments above and by generalized locality of $c$ it follows that
\begin{eqnarray*}
1_A c_{s,t}(\bar{D} \bar{Q}) &=& 1_A c_{s,t}(D^1 Q_1) \\
1_{A^c} c_{s,t}(\bar{D} \bar{Q}) &=& 1_{A^c} c_{s,t}(D^2 Q_2),
\end{eqnarray*}
then
\begin{eqnarray*}
&&\bar{D}_{s,t} \bE_{\bar{Q}} \left[\left. -X \right| \mathcal{F}_s \right]-c_{s,t}(\bar{D}\bar{Q})\\
&=& \bar{D}_{s,t} \left( 1_A \bE_{Q_1} \left[\left. -X \right| \mathcal{F}_s \right]+1_{A^c} \bE_{Q_2} \left[\left. -X \right| \mathcal{F}_s \right] \right) -1_A c_{s,t}(D^1 Q_1)-1_{A^c} c_{s,t} (D^2 Q_2) \\
&=& 1_A \left(\bar{D}_{s,t} \bE_{Q_1} \left[\left. -X \right| \mathcal{F}_s \right]-c_{s,t}(D^1 Q_1)\right) +1_{A^c} \left(\bar{D}_{s,t} \bE_{Q_2} \left[\left. -X \right| \mathcal{F}_s \right] - c_{s,t} (D^2 Q_2)\right) \\
&=& 1_A \left(D^1 _{s,t} \bE_{Q_1} \left[\left. -X \right| \mathcal{F}_s \right]-c_{s,t}(D^1 Q_1)\right) +1_{A^c} \left(D^2 _{s,t} \bE_{Q_2} \left[\left. -X \right| \mathcal{F}_s \right] - c_{s,t} (D^2 Q_2)\right) \\
&=& \max_{i=1,2} \left\{D^i_{s,t} \bE_{Q_i} \left[\left. -X \right| \mathcal{F}_s \right]-c_{s,t}(D^i Q_i) \right\}.
\end{eqnarray*}
Hence the set $\mathcal{C}_X$ is upward directed.

By Theorem A.32 of F\"{o}llmer and Schied \cite{follmers-book} (see also Detlefsen and Scandolo \cite{detlef-scandolo}), it follows that there exists an increasing sequence $\left\{ D^n _{s,t} \bE_{Q_n} \left[\left. -X \right| \mathcal{F}_s \right]-c_{s,t}(D^n Q_n) \right\}_n \subseteq \mathcal{C}_X$  such that $\esssup \mathcal{C}_X = \lim _{n \to +\infty} \{ D^n _{s,t} \bE_{Q_n} \left[\left. -X \right| \mathcal{F}_s \right]-c_{s,t}(D^n Q_n) \}$. Hence
\begin{eqnarray}
&& \rho_{s,t} \left(-\rho_{t,u} (X) \right) \\
&=& \esssup _{D \in \mathcal{D}, Q \in \mathcal{Q}} \left\{D_{s,t} \bE_Q \left[\left. \rho_{t,u}(X) \right| \mathcal{F}_s \right] -c_{s,t}(DQ) \right\} \notag \\
&=& \lim_n \left\{ D^n _{s,t} \bE_{Q_n} \left[\left. \rho_{t,u}(X) \right| \mathcal{F}_s \right] -c_{s,t}(D^n Q_n) \right\} \notag \\
&=& \lim_n \left\{D^n _{s,t} \bE_{Q_n} \left[\left. \lim_k \left\{D^k _{t,u} \bE_{Q_k} \left[\left. -X \right| \mathcal{F}_t \right] -c_{t,u}(D^k Q_k) \right\} \right| \mathcal{F}_s \right] -c_{s,t}(D^n Q_n) \right\} \notag \\
&=& \lim_n \left\{D^n _{s,t} \lim_k \left\{ \bE_{Q_n} \left[\left. D^k _{t,u} \bE_{Q_k} \left[\left. -X \right| \mathcal{F}_t \right] -c_{t,u}(D^k Q_k) \right| \mathcal{F}_s \right]\right\} -c_{s,t}(D^n Q_n) \right\} \notag \\
&=& \lim_n \lim_k \left\{D^n _{s,t}  \bE_{Q_n} \left[\left. D^k _{t,u} \bE_{Q_k} \left[\left. -X \right| \mathcal{F}_t \right]\right| \mathcal{F}_s \right] - D^n _{s,t} \bE_{Q_n} \left[ \left. c_{t,u}(D^k Q_k) \right| \mathcal{F}_s \right] -c_{s,t}(D^n Q_n) \right\} \notag \\
&=& \lim_n \lim_k \left\{ D^n _{s,t}  \bE_{Q_n} \left[\left. D^k _{t,u} \bE_{Q_k} \left[\left. -X \right| \mathcal{F}_t \right] \right| \mathcal{F}_s \right] -c_{s,u}(D^{n,k} Q_{n,k}) \right\} \label{eq: eq10001} \\
&=& \lim_n \lim_k \left\{  D^{n,k} _{s,u} \bE_{Q_{n,k}} \left[\left. -X \right| \mathcal{F}_s \right]  -c_{s,u}(D^{n,k} Q_{n,k}) \right\} \label{eq: eq 10002} \\
& \leq & \rho_{s,u}(X), \notag
\end{eqnarray}
where \eqref{eq: eq10001} and \eqref{eq: eq 10002} are due to the generalized cocycle (Cc) of $c$ and to the stability for joint pasting (QDa), and $D^{n,k}$ and $Q_{n,k}$ denote the pasting between $n$ and $k$-versions.

It remains to prove the converse inequality, that is $\rho_{s,t} \left(-\rho_{t,u} (X) \right) \geq \rho_{s,u}(X)$. Proceeding as previously, we get
\begin{eqnarray}
&&\rho_{s,u} (X) \\
&=& \esssup _{D \in \mathcal{D}, Q \in \mathcal{Q}} \left\{D_{s,u} \bE_Q \left[\left. -X \right| \mathcal{F}_s \right] -c_{s,u}(DQ) \right\} \notag \\
&=& \lim_m \left\{D^m _{s,u} \bE_{Q_m} \left[\left. -X \right| \mathcal{F}_s \right] -c_{s,u}(D^m Q_m) \right\} \notag \\
&=& \lim_m \left\{D^m _{s,u} \bE_{Q_m} \left[\left. -X \right| \mathcal{F}_s \right] -c_{s,t}(D^m Q_m)-\bE_{Q_m} \left[ \left. D^m _{s,t} c_{t,u} (D^m Q_m) \right| \mathcal{F}_s  \right] \right\} \notag \\
&=& \lim_m \left\{D^m _{s,t} \bE_{Q_m} \left[ \left. D^m _{t,u} \bE_{Q_m} \left[\left. -X \right| \mathcal{F}_t \right] \right| \mathcal{F}_s \right] -c_{s,t}(D^m Q_m)-\bE_{Q_m} \left[ \left. D^m _{s,t} c_{t,u} (D^m Q_m) \right| \mathcal{F}_s  \right] \right\} \label{eq: pasting D and Q} \\
&=& \lim_m \left\{D^m _{s,t} \bE_{Q_m} \left[ \left. D^m _{t,u} \bE_{Q_m} \left[\left. -X \right| \mathcal{F}_t \right]-c_{t,u} (D^m Q_m) \right| \mathcal{F}_s \right] -c_{s,t}(D^m Q_m)\right\} \notag \\
&\leq & \lim_m \left\{D^m _{s,t} \bE_{Q_m} \left[ \left. \rho_{t,u} (X) \right| \mathcal{F}_s \right] -c_{s,t}(D^m Q_m) \right\} \notag \\
& \leq & \rho_{s,t} \left(-\rho_{t,u} (X) \right), \notag
\end{eqnarray}
where \eqref{eq: pasting D and Q} is due to assumption (QDa).
This concludes the proof.


\begin{thebibliography}{99}

  \bibitem{AcPe11}
	Acciaio, B., and Penner, I.: Dynamic risk measures. Advanced Mathematical Methods for Finance. Springer Berlin Heidelberg. pgg. 1--34 (2011).

\bibitem{ArDeEbHe99} Artzner, P., Delbaen, F., Eber, J.M., Heath, D.: Coherent measures
of risk, Mathematical Finance, 9(3), 203-228 (1999).

  \bibitem{ArDeEbHeKu04} Artzner, P., Delbaen, F., Eber, J.M., Heath, D., Ku, H.: Coherent
multiperiod risk adjusted values and Bellman's principle. Preprint, ETH Zurich, Switzerland (2004).

\bibitem{bion-nadal-FS} Bion-Nadal, J.: Dynamic risk measures: Time consistency and risk measures from BMO martingales. Finance and Stochastics 12, 219-244 (2008).

\bibitem{bion-nadal-SPA} Bion-Nadal, J.: Time consistent dynamic risk processes. Stochastic Processes and their Applications 119, 633-654 (2009).

\bibitem{CMMM2}  Cerreia-Vioglio, S., Maccheroni, F., Marinacci, M., Montrucchio, L.:
Risk measures: rationality and diversification, Mathematical Finance 21/4, 743-774 (2011).

\bibitem{CDK1}  Cheridito, P., Delbaen, F., Kupper, M.: Coherent and
convex monetary risk measures for bounded c\`{a}dl\`{a}g processes.
Stochastic Processes and their Applications 112/1, 1-22 (2004).

\bibitem{CDK2}  Cheridito, P., Delbaen, F., Kupper, M.: Coherent and
convex monetary risk measures for unbounded c\`{a}dl\`{a}g
processes. Finance and Stochastics 9/3, 369-387 (2005).

\bibitem{CDK3} Cheridito, P., Delbaen, F., Kupper, M.:
Dynamic monetary risk measures for bounded discrete time processes.
Electronic Journal of Probability 11, 57-106 (2006).

\bibitem{CK} Cheridito, P., Kupper, M.: Recursiveness of indifference prices
and translation-invariant preferences. Mathematics and Financial Economics 2, 173-188 (2009).

\bibitem{De00} Delbaen, F.: Coherent risk measures.
Cattedra Galileiana. Lecture notes, Scuola Normale Superiore, Pisa, Italy (2000).

\bibitem{delb}  Delbaen, F.: Coherent Risk Measures on General Probability
Spaces, in: Advances in Finance and Stochastics, K. Sandmann and P.J. Sch\"{o}nbucher eds., Springer-Verlag, 1-37 (2002).

\bibitem{delb-mstable} Delbaen, F.: The structure of $m$-stable sets and in
particular of the set of risk neutral measures. In: Memoriam
Paul-And\'{e} Meyer, Lecture Notes in Mathematics  1874, pp. 215-258 (2006).

\bibitem{DPRG} Delbaen, F., Peng, S., Rosazza Gianin, E.: Representation of the penalty term of dynamic concave utilities. Finance and Stochastics 14/3, 449-472 (2010).

\bibitem{detlef-scandolo} Detlefsen, K., Scandolo, G.: Conditional and dynamic
convex risk measures. Finance \& Stochastics 9/4, 539-561 (2005).

\bibitem{drapeau-kupper}  Drapeau, S., Kupper, M.: Risk Preferences
and their Robust Representation. Mathematics of Operations Research 38/1, 28-62 (2013).

\bibitem{drapeau-kupper-tangpi-rg} Drapeau, S., Kupper, M., Tangpi, L., Rosazza Gianin, E.: Dual Representation of Minimal Supersolutions of Convex BSDEs. Forthcoming on Annales de l'Institut Henri Poincar\'{e} (B) Probabilit\'{e}s et Statistiques (2013).

\bibitem{EkTe99} Ekeland, I., and Temam, R. Convex analysis and variational inequalities. Classics in Applied Mathematics. Society for Industrial and Applied Mathematics (1999).

\bibitem{ELK-rav} El Karoui, N., Ravanelli, C.: Cash sub-additive risk measures and interest rate ambiguity, Mathematical Finance 19, 561-590 (2009).

\bibitem{follmer-penner} F{\"o}llmer, H., Penner, I.: Convex risk measures and
the dynamics of their penalty functions. Statistics and Decisions
24/1, 61-96 (2006).

\bibitem{follers}  F{\"o}llmer, H., Schied, A.: Convex measures of risk and
trading constraints. Finance \& Stochastics 6/4, 429-447 (2002).

\bibitem{follmers-book} F\"{o}llmer, H., Schied, A.: Stochastic Finance, 2nd edition, de Gruyter (2008).

\bibitem{fritt-mag-dyn}  Frittelli, M., Maggis, M.: Dual representation of
quasiconvex conditional maps, SIAM Journal of
Financial Mathematics 2, 357--382 (2011).

\bibitem{frittelli-rg} Frittelli, M., Rosazza Gianin, E.: Putting order in risk
measures. Journal of Banking and Finance 26/7, 1473-1486 (2002).

\bibitem{frittelli-rg2}  Frittelli, M., Rosazza Gianin, E.: Dynamic Convex Risk
Measures. In: Szeg\"{o}, G. (ed.), Risk Measures for the 21st
Century, J. Wiley, 227-248 (2004).

\bibitem{kloppel-schweizier}  Kl{\"o}ppel, S., Schweizer, M.: Dynamic utility
indifference valuation via convex risk measures. Mathematical
Finance 17/4, 599-627 (2007).

\bibitem{Ri04} Riedel F.: Dynamic coherent risk measures. Stochastic Processes
and their Applications, 112 (2), 185-200 (2004).

\bibitem{RoSch15} Roorda, B., Schumacher, J.M.: Weakly time consistent concave valuations and their dual representations. Forthcoming on Finance and Stochastics.

\bibitem{ros}  Rosazza Gianin, E.: Risk measures via
g-expectations. Insurance: Mathematics and Economics 39, 19-34 (2006).

\end{thebibliography}
\end{document}